\documentclass[reqno]{amsart}
\usepackage{enumerate,amsmath,amssymb, color}
\usepackage{hyperref}
\usepackage{a4wide}

\usepackage[english]{babel}
\usepackage[utf8]{inputenc}

\usepackage[dvipsnames]{xcolor}

\newcommand{\abs}[1]{\left\lvert#1\right\rvert}				

\renewcommand{\a}{\alpha}
\newcommand{\g}{\gamma}
\renewcommand{\o}{\omega}

\renewcommand{\b}{\beta}
\renewcommand{\d}{\partial}

\newcommand{\ldr}[1]{\langle #1\rangle}

\newcommand{\R}[0]{\mathbb{R}}							
\newcommand{\C}[0]{\mathbb{C}}							

\theoremstyle{plain}
\newtheorem{thm}{Theorem}[section]

\newtheorem{lemma}[thm]{Lemma}
\newtheorem{cor}[thm]{Corollary}
\theoremstyle{definition}
\newtheorem{definition}[thm]{Definition}

\newtheorem{remark}[thm]{Remark}

\newtheorem{mainthm}{Main Result}

\author{Oliver Lindblad Petersen}
\title{The mode solution of the wave equation in Kasner spacetimes and redshift}
\address{Universität Potsdam, Institut für Mathematik, Karl-Liebknecht-Str. 24-25, 14476 Potsdam OT Golm, Germany}
\urladdr{http://www.math.uni-potsdam.de/professuren/geometrie/personen/oliver-lindblad-petersen/}
\email{lindblad@uni-potsdam.de}

\begin{document}
\hbadness=100000
\vbadness=100000

\begin{sloppypar}

\begin{abstract}
We study the mode solution to the Cauchy problem of the scalar wave equation $\Box \varphi = 0$ in Kasner spacetimes. As a first result, we give the explicit mode solution in axisymmetric Kasner spacetimes, of which flat Kasner spacetimes are special cases. Furthermore, we give the small and large time asymptotics of the modes in general Kasner spacetimes. Generically, the modes in non-flat Kasner spacetimes grow logarithmically for small times, while the modes in flat Kasner spacetimes stay bounded for small times. For large times, however, the modes in general Kasner spacetimes oscillate with a polynomially decreasing amplitude. This gives a notion of large time frequency of the modes, which we use to model the wavelength of light rays in Kasner spacetimes. We show that the redshift one obtains in this way actually coincides with the usual cosmological redshift.
\end{abstract}


\maketitle
\tableofcontents

\section{Introduction}

Let throughout the paper  
\begin{align*}
	M :=& \R_+ \times \R^3, \\
	g :=& -dt^2 + \sum_{j=1}^3 t^{2p_j} (dx^j)^2,
\end{align*}
where
\begin{align}
	\sum_{j=1}^3 {p_j}^2 &= 1, \label{eq: P-L2} \\
	\sum_{j=1}^3 p_j &= 1. \label{eq: P-L1}
\end{align}
$(M, g)$ is called a \emph{Kasner spacetime}. 
The Kasner spacetimes are the non-trivial Bianchi type I spacetimes satisfying Einstein's vacuum equations, i.\,e. ${\mathrm{Ric}(g) = 0}$. 
One observes that Kasner spacetimes are spatially anisotropic, they might "expand" in one direction and "shrink" in another direction. It is clear that at least one $p_j$ must satisfy
\[		0 \leq p_j \leq 1.
\]
If some $p_j = 1$, then the other two must equal $0$. In this case, the Kasner spacetime is flat. If all $p_j \neq 1$, then all $p_j$ are non-zero and the Kasner spacetime is non-flat. (See also \cite[page 197]{Stephani}.) The purpose of this paper is to study solutions to the wave equation $\Box \varphi = 0$ on Kasner spacetimes. Given a smooth function $\varphi$ with spatially compact support, define 
\[
	\a_\o(t):= F(\varphi(t,\cdot))(\o),
\]
where $F$ is the Fourier transform on $\{t \} \times \R^3$ and $\o \in \R^3$. Since the metric only depends on time, the wave equation
\[
	\Box \varphi = 0
\]
transforms into linear second order ODE's for $\a_\o$. We call $\a_\o$ a \emph{mode solution} to the wave equation. In case of axisymmetry, $\a_\o$ can be written down explicitly. 

\begin{mainthm}[Explicit solutions] \label{mainthm: ExplicitSolutions}
If the Kasner spacetime is \emph{axisymmetric}, i.\,e.\,if two of the $p_i$'s are equal, then $\a_\o$ can be given explicitly in terms of Bessel functions and Heun Biconfluent functions. 
\end{mainthm}

See Theorem \ref{thm: Explicit1} and Theorem \ref{thm: Explicit2} for the precise statement. This extends results of \cite{ScalarWaveEquation}. In general, it seems like there is no easy way to write down the solutions explicitly. We therefore study the asymptotics of $\a_\o$ for $t \to 0$ and $t \to \infty$.

\begin{mainthm}[Small time asymptotics] \label{mainthm: SmallTimes}
If the Kasner spacetime is non-flat, then there are constants $c_1, c_2 \in \C$, depending on $\o$, such that
\[
 \a_\o(t) - \left(c_1 \ln(t) + c_2 \right) \to 0,
\]
as $t \to 0$. If it is flat, then for a generic $\o = (\o_1, \o_2, \o_3) \in \R^3$, there are constants $c_1, c_2 \in \C$, depending on $\o$, such that
\[
  \a_\o(t) - \left(c_1 e^{2 \pi i \o_1 \ln(t)} + c_2 e^{- 2 \pi i \o_1 \ln(t)}\right) \to 0,
\]
as $t \to 0$. 
\end{mainthm}

\begin{mainthm}[Large time asymptotics] \label{mainthm: LargeTimes}
For any $t_0 > 0$, there are constants $c_1, c_2 \in \C$, depending on $\o$, such that
\[
 \a_\o(t) \left(\sum_{j=1}^3 {\o_j}^2t^{2-2p_j} \right)^{1/4} - \left(c_1 e^{2 \pi i \int_{t_0}^t f_\o(u)du} + c_2 e^{-2 \pi i \int_{t_0}^t f_\o(u)du}\right) \to 0,
\]
as $t \to \infty$, where $f_\o(t) := \left( \sum_{j=1}^3 \frac{{\o_j}^2}{t^{2p_j}} \right)^{1/2}$.
\end{mainthm}

In particular, we deduce bounds on the amplitudes:
\[
 \abs{\a_\o(t)} \leq \frac {\abs{c_1} + \abs{c_2} + 1} {\left(\sum_{j=1}^3 {\o_j}^2t^{2-2p_j} \right)^{1/4}}
\]
for large times $t$. The proofs of Main Result \ref{mainthm: SmallTimes} and \ref{mainthm: LargeTimes} are presented in Section \ref{section: SmallTimes} and \ref{section: LargeTimes}.

We conclude that the modes start to oscillate for large times $t$ and we call $f_\o$ the \emph{large time frequency} for the mode $\a_\o$. Interpreting this as the frequency of a light ray with linear momentum $\o \in \R^3$, we get a notion of \emph{large time redshift}. As an application, we show that the large time redshift coincides with the classical notion of redshift, defined to be inverse proportional to the energy of lightlike geodesics. 

If $p_1 = p_2 = p_3 = \frac23$, then $(M,g)$ is called an Einstein-de Sitter spacetime (this is not a Kasner metric). This spacetime has a singularity structure at $t=0$ similar to the non-flat Kasner spacetime. In \cite[Theorem 1.1]{Galastian}, solutions to the wave equation are given that grow like $\frac 1t$ for $t \to 0$. Hence they grow faster than the modes in Main Result \ref{mainthm: SmallTimes}, which grow like $\ln(t)$ for $t \to 0$. 

The wave equation $\square \varphi = 0$ has been studied on various backgrounds like Robertson-Walker spacetimes, see e.\,g.\,\cite{AbbasiCraig,Ariane,KlainermanSarnak} and in de Sitter and anti-de Sitter spacetimes, see e.\,g.\,\cite{Polarski}. It has also been studied extensively on black hole spacetimes, see for example \cite{DafermosRodnianski3}. However, results on the wave equation in Kasner spacetimes seem hard to find in the literature. The only reference known to the author is \cite{ScalarWaveEquation}, where certain integral operators for the solution are constructed.

The paper is structured as follows. In Section \ref{sec: Kasner}, we define mode solutions in Kasner spacetimes. Section \ref{sec: Explicit} is devoted for the explicit solutions in axisymmetric Kasner spacetimes. The small and the large time asymptotics are presented in Section \ref{sec: SmallTimes} and Section \ref{sec: LargeTimes}. The application to redshift is given in Section \ref{sec: Redshift}.
\vspace{0.3cm}

\noindent
\textbf{Acknowledgements}
The author would like to thank Christian B\"ar for introducing him to the topic of mode solutions of wave equations and Hans Ringstr\"om for introducing him to Kasner spacetimes. He especially would like to thank Andreas Hermann for many helpful discussions. Moreover, he would like to thank the Berlin Mathematical School and Sonderforschungsbereich 647 funded by Deutsche Forschungsgemeinschaft for financial support.

\section{The wave equation in Kasner spacetimes} \label{sec: Kasner}

The d'Alembert operator in a Kasner spacetime is given by
\[
 \square = \d_t^2 + \frac1{t}\d_t - \sum_{j=1}^3 \frac{1}{t^{2p_j}} \d_j^2.
\]
In order to fix the notation, let us define the Fourier transform.

\begin{definition}[Fourier transform]
For $f \in L^1(\R^n)$, define the \emph{Fourier transform of $f$} as
\[
 F(f)(\omega) := \int_{\R^n} f(x)e^{2\pi i x \cdot \o}dx,
\]
for all $\omega \in \R^3$.
\end{definition}

The solution to the wave equation is decomposed into Fourier modes as follows.
Let $t_0 \in \R_+$ and $\varphi_0, \varphi_n \in C_c^\infty(\{t_0\} \times \R^3)$ be given. By \cite[Theorem 3.2.11]{Baer} there is a unique solution $\varphi: M \to \R$, solving the Cauchy problem
\[	\square \varphi = 0, \ \varphi |_{\{t_0\} \times \R^3} = \varphi_0, \ \d_t\varphi |_{\{t_0\} \times \R^3} = \varphi_n. 
\]
Let us write the solution as
 \[
  \varphi(t,x) = \int_{\R^3} \a_\o(t)e^{-2\pi i x \cdot \o}d\o
 \]
for all $(t,x) \in M$, where $\a_\o:\R_+ \rightarrow \C$ is the unique solution to 
\begin{align}
	&\a_{\o}''(t) + \frac{\a_\o'(t)}{t} + \a_\o(t) 4 \pi^2 \sum_{j=1}^3 \frac{\o_j^2}{t^{2p_j}} = 0, \ \forall t \in \R_+,\label{eq: AlphaODE}		\\
	&\a_\o(t_0) = \int_{\R^3} \varphi_0(x) e^{2\pi i \o \cdot x} dx, \label{eq: AlphaInitialValue}	\\
	&\a'_\o(t_0) = \int_{\R^3} \varphi_n(x) e^{2\pi i \o \cdot x} dx.\label{eq: AlphaInitialDerivative}
\end{align}

\begin{definition}
The set of all $\a_\o$ in the above lemma, i.\,e.\,the solutions to (\ref{eq: AlphaODE} - \ref{eq: AlphaInitialDerivative}) for different $\o \in \R^3$, is called the \emph{mode solution to the wave equation in Kasner spacetimes}. For a fixed $\o \in \R^3$, the solution $\a_\o$ is called a \emph{mode}.
\end{definition}

From now on, we will study the modes $\a_\o$ for a fixed $\o \in \R^3$. It is convenient to rewrite the ODE's as follows. 

\begin{lemma} \label{le: AlphaToBeta}
The solution $\a_\o$ in the previous theorem can be written as
\[
  \a_\o(t) =: \b_\o(\ln(t))
\]
where $\b_\o: \R \to \C$ is the unique solution to
\begin{align}
	&\b_\o''(s) + \b_\o(s) K_\omega(s) = 0 \label{eq: BetaODE}, \\
	&\b_\o(\ln(t_0)) = \int_{\R^3} \varphi_0(x) e^{2\pi i \o \cdot x} dx, \label{eq: BetaInitialValue}	\\
	&\b'_\o(\ln(t_0)) = t_0 \int_{\R^3} \varphi_n(x) e^{2\pi i \o \cdot x} dx.\label{eq: BetaInitialDerivative}
\end{align}
where $K_\omega(s) := 4 \pi^2 \sum_{j=1}^3 {\o_j}^2 e^{(2-2p_j)s}$.
\end{lemma}

\section{The explicit modes in axisymmetric Kasner spacetimes} \label{sec: Explicit}

Let us give the explicit solutions to equations (\ref{eq: AlphaODE} - \ref{eq: AlphaInitialDerivative}) for axisymmetric Kasner spacetimes. This subclass includes the flat Kasner spacetimes. Up to permutation of the indices, there are two possibilities. 

\begin{remark}
If $p = (p_1,p_2,p_3)$ satisfies \eqref{eq: P-L2}, \eqref{eq: P-L1} and two of the $p_i$ are equal, then
\begin{align*}
 &\{p_1, p_2, p_3\} = \{1,0,0\}  \quad \text{(flat Kasner metric)} \\
 \text{ or } \quad &\{p_1, p_2, p_3\} = \left\{-\frac{1}{3}, \frac{2}{3}, \frac{2}{3} \right\}.
\end{align*}
\end{remark}

\subsection{The explicit mode solution in flat Kasner spacetimes}

Choose
\[	
  p_1=1, \quad p_2 = p_3 = 0,
\]
the other cases are analogous. Equation \eqref{eq: AlphaODE} becomes
\[
  \a_{\o}''(t) + \frac{\a_{\o}'(t)}{t} + 4 \pi^2 \a_{\o}(t) \left({\o_1}^2t^{-2} + {\o_2}^2 + {\o_3}^2\right) = 0,
\]
for all $t \in \R_+$. We will see that for generic $\o \in \R^3$, the solution can be expressed by Bessel functions. The Bessel functions of first and second kind, usually denoted by $J_\nu, Y_\nu:\C \to \C$ for $\nu \in \C$ are linearly independent solutions to the \emph{Bessel equation}
\begin{equation} \label{def: Bessel}
	x^2y''(x) + xy'(x) + (-\nu^2+x^2)y(x) = 0.
\end{equation}
See e.\,g.\,\cite{AbramowitzStegun} for the definitions and basic properties of Bessel functions. 

The solution for ${\o_2}^2 + {\o_3}^2 \neq 0$, one can find in \cite[Section II]{ScalarWaveEquation}. We complete the list by adding the solution when $\o_2=\o_3= 0$.

\begin{thm}[The explicit solution in flat Kasner spacetimes] \label{thm: Explicit1}
Let $(M,g)$ be a flat Kasner spacetime with $p_1 = 1, p_2 = p_3 = 0$. The solution to equations (\ref{eq: AlphaODE} - \ref{eq: AlphaInitialDerivative}) is given by the following, for the different cases of $\o = (\o_1, \o_2, \o_3) \in \R^3$:
\begin{itemize}
 \item[•] $\o_1 = \o_2 = \o_3 = 0$:
\[
  \a_\o (t) = \a_\o'(t_0) \ln\left(\frac{t}{t_0}\right)t_0 + \a_\o(t_0),
\]
\item[•] $\o_1 \neq 0, \ \o_2 = \o_3 = 0$:
\[
  \a_\o (t) = c_1 e^{2\pi i \o_1 \ln(t)} + c_2 e^{-2 \pi i \o_1 \ln(t)},
\]
 \item[•] ${\o_2}^2 + {\o_3}^2 \neq 0$:
\[	
  \a_\o(t) = c_1 J_{2i \pi \o_1} \left(2 \pi t \sqrt{\o_2^2 + \o_3^2} \right) + c_2 Y_{2 i \pi \o_1}\left( 2\pi t \sqrt{\o_2^2 + \o_3^2} \right),
\]
\end{itemize}
where $c_1, c_2 \in \C$ are constants depending on the initial data given in \eqref{eq: AlphaInitialValue} and \eqref{eq: AlphaInitialDerivative}.
\end{thm}

\begin{proof}
The proof is a simple verification.
\end{proof}

\subsection{The explicit mode solution in axisymmetric non-flat Kasner spacetimes}

Let
\[
  p_1 = -\frac{1}{3}, \quad p_2 = p_3 = \frac{2}{3},
\]
the other cases are analogous. Equation \eqref{eq: AlphaODE} becomes
\[
	\a_{\o}''(t) + \frac{\a_{\o}'(t)}{t} + 4 \pi^2 \a_{\o}(t) \left({\o_1}^2t^{2/3} + ({\o_2}^2 + {\o_3}^2)t^{-4/3}\right) = 0,
\]
for all $t \in \R_+$. 

In addition to the Bessel functions, the solution will be expressed in terms of a so called "Heun Biconfluent function". The definition of the Heun Biconfluent functions follows the documentation of the computer algebra system Maple, which in turn refers to \cite{Heun}. See also \cite{SpectralHeun} for a discussion of the Heun Biconfluent function. The literature list of \cite{SpectralHeun} includes several examples from physics and mathematics where this function is used and its properties discussed.

\begin{definition}[Heun Biconfluent function] \cite{Heun}
The \emph{Heun Biconfluent function} $\mathrm{HeunB}(\a, \b, \g, \delta, \cdot ): \C \to \C$ of the fixed constants $\a, \b, \g, \delta \in \C$, $\a \neq -1$ is defined as the unique solution to
\[
  xy''(x) - \left(\b x+ 2 x^2 - \a - 1\right) y'(x) - \frac{1}{2}((2 \a - 2 \g + 4) x + \b \a + \b + \delta)y(x) = 0,
\]
such that 
\[
  y(0) = 1, \quad y'(0) = \frac{\a \b + \b + \delta}{2 \a + 2}.
\]
\end{definition}

We need the Heun Biconcluent function only in the case where all constants vanish except $\delta$. In this case
\[
\mathrm{HeunB}(0, 0, 0, \delta, \cdot ): \C \to \C
\]
is the unique solution to
\[
  xy''(x) + \left(1 - 2x^2\right) y'(x) - \left(2 x + \frac{\delta}2 \right)y(x) = 0,
\]
such that 
\[
  y(0) = 1, \quad y'(0) = \frac \delta 2.
\]

\begin{thm}[The explicit solution in non-flat axisymmetric Kasner spacetimes] \label{thm: Explicit2}
Let $(M,g)$ be a (non-flat) Kasner spacetime with $p_1 = - \frac{1}{3}, p_2 = p_3 = \frac{2}{3}$. The solution to equations (\ref{eq: AlphaODE} - \ref{eq: AlphaInitialDerivative}) is given by the following for the different cases of $\o = (\o_1, \o_2, \o_3) \in \R^3$:
\begin{itemize}
 \item[•] $\o_1 = \o_2 = \o_3 = 0$:
\[
  \a_\o (t) = \a_\o'(t_0) \ln\left(\frac{t}{t_0}\right)t_0 + \a_\o(t_0),
\]
 \item[•] $\o_1 \neq 0, \ \o_2=\o_3 = 0 $:
\[
  \a_\o (t) = c_1 J_0\left(\frac{3}{2}\pi \o_1 t^{4/3}\right) + c_2Y_0\left(\frac{3}{2}\pi \o_1 t^{4/3} \right),
\]
 \item[•] $\o_1 = 0, \ {\o_2}^2 + {\o_3}^2 \neq  0$:
\[		
  \a_\o (t) = c_1J_0 \left(6\pi \sqrt{{\o_2}^2 + {\o_3}^2}t^{1/3}\right) + c_2Y_0 \left(6\pi \sqrt{{\o_2}^2 + {\o_3}^2}t^{1/3}\right),
\]
 \item[•] $\o_1 \neq 0, \ {\o_2}^2 + {\o_3}^2 \neq 0$:
\begin{align*}
	\a_\o(t) = & e^{\frac32\pi i \abs{\o_1} t^{4/3}} \mathrm{HeunB} \left(0,0,0, \delta_\o, L_\o t^{2/3} \right) \\
	& \cdot \left( \int_{t_0}^t \frac{e^{-3\pi i \abs{\o_1} u^{4/3}}}{\mathrm{HeunB} \left(0,0,0, \delta_\o, L_\o u^{2/3} \right)^2 u} du \cdot c_1 + c_2 \right),
\end{align*}
\end{itemize}
where 
\begin{align*}
  &\delta_\o := -18 \pi^2 \frac{{\o_2}^2 + {\o_3}^2}{L_\o}, \\
  &L_\o := \sqrt{\frac32\pi \abs{\o_1}} (1+i),
\end{align*}
and $c_1, c_2 \in \C$ are constants depending on initial data \eqref{eq: AlphaInitialValue} and \eqref{eq: AlphaInitialDerivative}.
\end{thm}
\begin{proof}
The proof is a straightforward verification.
\end{proof}

\section{Small time asymptotics of the modes in general Kasner spacetimes} \label{section: SmallTimes} \label{sec: SmallTimes}

\begin{thm}[The asymptotics for small times] \label{thm: SmallTimes}
Let $(M,g)$ be a Kasner spacetime and let $\o \in \R^3$. Assume first that if some $p_j = 1$, then $\o_j = 0$. Then there exist constants $c_1, c_2 \in \C$, depending on $\o$, such that
\[
 \a_\o(t) - (c_1 \ln(t) + c_2) \to 0
\]
as $t \to 0$. If $p_j = 1$ and $\o_j \neq 0$, then
\[
  \a_\o(t) - \left(c_1 e^{2 \pi i \o_j \ln(t)} + c_2 e^{- 2 \pi i \o_j \ln(t)} \right) \to 0
\]
as $t \to 0$, for some constants $c_1, c_2 \in \C$ depending on $\o$.
\end{thm}

\begin{remark}[Optimal bound on the growth for small times]
We claim that the constants $c_1$ and $c_2$ in the previous theorem cannot be set to zero in general. Theorem \ref{thm: Explicit2} implies that if $p_1 = -\frac13$ and $p_2 = p_3 = \frac23$ and $\o_2 = \o_3 = 0$, the solution is given by
\[
		\a_\o (t) = \tilde c_1 J_0\left(\frac{3}{2}\pi \o_1 t^{4/3}\right) + \tilde c_2 Y_0\left(\frac{3}{2}\pi \o_1 t^{4/3} \right),
\]
where $\tilde c_1, \tilde c_2 \in \C$. If we choose initial data so that $\tilde c_1 = 0 \neq \tilde c_2$, then $\abs{\a_\o(t)} \to \infty$ as $t \to 0$, so we cannot set $c_1 = 0$ in general. If we choose the initial data so that $\tilde c_1 \neq 0 = \tilde c_2$, then $\a_\o(t) \to \tilde c_1$ as $t \to 0$, so $c_2$ does not vanish in general. An analogous argument holds in the flat Kasner case, using Theorem \ref{thm: Explicit1}. Hence the result of Theorem \ref{thm: SmallTimes} is optimal.
\end{remark}

We proceed by proving the theorem.

\begin{proof}[of Theorem \ref{thm: SmallTimes}]
Since the statement is trivial if $\o = 0$, let us assume that $\o \neq 0$. Assume first that $(M,g)$ is a non-flat Kasner spacetime or that $(M,g)$ is a flat Kasner spacetime with $\o_j = 0$ for the index $j$ such that $p_j = 1$. Recall that by Lemma 1, equation \ref{eq: AlphaODE} is equivalent to equation \ref{eq: BetaODE}. The above assumptions are exactly what is needed for $K_\o(s)$ to decay exponentially as $s \to -\infty$. The crucial observation in the proof is that
\begin{align*}
 \frac d{ds}\left( \b_\o'(s)^2 + K_\o(s)\b_\o(s)^2 \right) 
  &= 2\b_\o'(s) \big( \b_\o''(s) + K_\o(s) \b_\o(s) \big) + K_\o'(s)\b_\o(s)^2 \\
  &= K_\o'(s)\b_\o(s)^2 \geq 0,
\end{align*}
since $K_\o$ is everywhere increasing. Hence
\begin{equation} \label{eq: EnergyEstimate}
 \b_\o'(s_1)^2 + K_\o(s_1)\b_\o(s_1)^2 \leq \b_\o'(s_2)^2 + K_\o(s_2)\b_\o(s_2)^2
\end{equation}
for any $s_1 \leq s_2$. Fix $s_c \in \R$, and define
\[
 C := \b_\o'(s_c)^2 + K_\o(s_c)\b_\o(s_c)^2.
\]
The estimate \eqref{eq: EnergyEstimate} implies that
\[
 K_\o(s)\b_\o(s)^2 \leq C
\]
for all $s \leq s_c$ and hence
\[
 \abs{\b_\o''(s)} \leq K_\o(s)\abs{\b_\o(s)} \leq \sqrt{K_\o(s)} \sqrt C,
\]
for all $s \leq s_c$. Since $K_\o$ decays exponentially for $s \to -\infty$, this implies that $\b_\o'' \in L^1(-\infty, s_c)$. Therefore, as $s \to - \infty$,
\[
 \b_\o'(s) = - \int_{s}^{s_c} \b_\o''(u) du + \b_\o'(s_c) \to -\int_{-\infty}^{s_c} \b_\o''(u)du + \b_\o'(s_c) =: c_1 \in \R.
\]

We now claim that there exists a $c_2 \in \R$ such that
\[
 \b_\o(s) - s c_1 \to c_2
\]
as $s \to -\infty$. To see this, note that
\[
 \b_\o'(s) - c_1 = \int_{-\infty}^s \b_\o''(u)du \leq \int_{-\infty}^s \abs{\b_\o''(u)}du \leq \sqrt C \int_{-\infty}^s \sqrt{K_\o(u)}du
\]
decays exponentially as $s \to -\infty$. In particular,
\[
 \b_\o'(s) - c_1 \in L^1(-\infty, s_c)
\]
and hence
\begin{align*}
 \b_\o(s) - c_1s
  &= -\int_{s}^{s_c} \b_\o'(u)du + \b_\o(s_c) - c_1s \\
  &= -\int_{s}^{s_c} \big[\b_\o'(u) - c_1 \big]du - c_1s_c + \b_\o(s_c)\\
  &\to -\int_{-\infty}^{s_c} \big[\b_\o'(u) - c_1 \big]du - c_1s_c + \b_\o(s_c) =: c_2,
\end{align*}
as $s \to -\infty$.

Reformulating the convergence using 
\[
 \b_\o(\ln(t)) = \a_\o(t)
\]
implies that
\[
 \a_\o(t) - c_1 \ln(t) \to c_2
\]
as $t \to 0$.

For the second statement, it is enough to check the result in Theorem \ref{thm: Explicit1} and to recall the small time asymptotics for Bessel functions with imaginary parameter $\nu$, see e.\,g.\,\cite{AbramowitzStegun}.
\end{proof}

\section{Large time asymptotics of the modes in general Kasner spacetimes}\label{section: LargeTimes} \label{sec: LargeTimes}

\begin{thm}[The asymptotics for large times] \label{thm: LargeTimes}
Let $(M,g)$ be a Kasner spacetime and let $\o \neq 0 \in \R^3$. Then there exist constants $c_1, c_2 \in \C$, depending on $\o$, such that
\[
 \a_\o(t) \left(\sum_{j=1}^3 {\o_j}^2t^{2-2p_j} \right)^{1/4} - \left[ c_1 e^{2 \pi i \int_{t_0}^t f_\o(u)du} + c_2 e^{-2 \pi i \int_{t_0}^t f_\o(u)du} \right] \to 0,
\]
 as $t \to \infty$, where
\[
f_\o(t) := \left(\sum_{j=1}^3 \frac{{\o_j}^2}{t^{2p_j}}\right)^{1/2}.
\]
In particular, there exists a $T \geq 0$ such that
\begin{equation} \label{eq: BoundAmplitude}
 \abs{\a_\o(t)} \leq \frac{\abs{c_1} + \abs{c_2} + 1}{\left(\sum_{j=1}^3 {\o_j}^2 t^{2-2p_j}\right)^{1/4}}
\end{equation}
for all $t \geq T$.
\end{thm}

Theorem \ref{thm: LargeTimes} implies that the amplitudes of the modes are monotone decreasing, but also that the modes start to oscillate as trigonometric functions. For trigonometric functions, we have a natural definition of frequency. We define this frequency as the `large time frequency` of the modes. 

\begin{definition}[Large time frequency of $\a_\o$] \label{def: LargeTimeFrequency}
We define the \emph{large time frequency of a mode $\a_\o$} to be $f_\o$ as in Theorem \ref{thm: LargeTimes}.
\end{definition}

\begin{cor}[Large time frequency in flat and non-flat Kasner spacetimes]\ 
\begin{itemize}
 \item[•] Let $(M, g)$ be a flat Kasner spacetime with $p_j =1$. Then 
 \[
  f_\o(t) \to \sqrt{{\o_k}^2 + {\o_l}^2}, \ \text{ as } t \to \infty, \ k,l \neq j.
 \]
 as $t \to \infty$.
 \item[•] Let $(M, g)$ be a non-flat Kasner spacetime, with $-\frac13 \leq p_k <0$. If $\o_k \neq 0$, then
 \[
  f_\o(t) \to \infty
 \]
 as $t \to \infty$, otherwise
 \[
 f_\o(t) \to 0
  \]
  as $t \to \infty$.
\end{itemize}
\end{cor}

It follows that the large time frequency, for a generic choice of $\o \in \R^3$, converges to a constant in flat Kasner spacetimes and goes to infinity in non-flat Kasner spacetimes. We proceed by proving the theorem. For this we will use the following lemma.

\begin{lemma} \label{le: MS10}

Assume that $K: (a,\infty) \to \R_+$ is smooth and $v:(a,\infty) \to \R$ satisfies
\begin{equation} \label{eq: ProofsComplexSolution}
  v''(s) + K(s)v(s) = 0, \quad \forall s \in (a, \infty).
\end{equation}
Assume furthermore that 
\[		
\psi_K := K^{-\frac{1}{4}} \frac{d^2}{ds^2} \left(K^{-\frac{1}{4}} \right) \in L^1(a,\infty)
\]
and that
\[
 K^{1/2} \notin L^1(a, \infty).
\]
 Then there exist solutions $v_1, v_2: (a, \infty) \to \C$ to equation \eqref{eq: ProofsComplexSolution} such that 
\begin{align*}
v_1(s)e^{-i\int_a^s K(u)^{1/2} du} K(s)^{1/4} \rightarrow 1 \\
v_2(s)e^{i\int_a^s K(u)^{1/2} du} K(s)^{1/4} \rightarrow 1
\end{align*}
when $s \rightarrow \infty$.
\end{lemma}
\begin{proof}
This is a special case of \cite[Proposition 3.2]{MS10}.
\end{proof}

We apply this result to our case.

\begin{proof}[of Theorem \ref{thm: LargeTimes}]
Recall, by Lemma \ref{le: AlphaToBeta}, that equation \eqref{eq: AlphaODE} is equivalent to equation \eqref{eq: BetaODE}. We claim that the assumptions in Lemma \ref{le: MS10} are satisfied with $K = K_\o$. Since $\o \neq 0$, note that $K_\o$ is nonzero and constant or grows exponentially as $s \to \infty$. If $K_\o$ is a nonzero constant, the assumptions in Lemma \ref{le: MS10} are trivially satisfied. Assume therefore that $K_\o$ is exponentially growing. Note that
\begin{align*}
 \psi_{K_\o}(s) 
    &= K_\o(s)^{-1/4} \frac{d^2}{ds^2} \left(K_\o(s)^{-1/4} \right) \\
    &= \frac{1}{4K_\o(s)^{1/2}} \left( \frac{5}{4} \frac{K_\o'(s)^2}{K_\o(s)^2} - \frac{K_\o''(s)}{K_\o(s)} \right).
\end{align*}
We claim that the term $\frac{K_\o''(s)}{K_\o(s)}$ is bounded. Let $k$ be an index such that $\o_k \neq 0$ and $p_k \leq p_j$ for all $j$ such that $\o_j \neq 0$. By factoring out $e^{(2-2p_k)s}$ from both numerator and denominator, we see that
\[
  \frac{K_\o''(s)}{K_\o(s)} = \frac{\sum_{j=1}^3 (2-2p_{j})^2 {\o_j}^2 e^{(2-2p_{j})s}}{\sum_{j=1}^3 {\o_j}^2 e^{(2-2p_{j})s}} = \frac{\sum_{j=1}^3 (2-2p_j)^2 {\o_j}^2 e^{2(p_k-p_j)s}}{\sum_{j=1}^3 \o_j^2 e^{2(p_k-p_j)s}}
\]
Since $p_k - p_j \leq 0$, we see that both denominator and numerator converge as $s \rightarrow \infty$. Furthermore, the denominator is bounded from below by ${\o_k}^2$. This implies that the quotient converges when $s \to \infty$. Analogously, one proves that the term $\frac{K_\o'(s)^2}{K_\o(s)^2}$
converges and it follows that
\[
\frac{5}{4} \frac{K_\o'(s)^2}{K_\o(s)^2} - \frac{K_\o''(s)}{K_\o(s)} 
\]
is bounded as $s \to \infty$. Hence it suffices to show that $K_\o^{-1/2} \in L^1(a, \infty)$, in order to prove that $\psi_{K_\o} \in L^1(a, \infty)$. But this is clear, since $K_\o$ grows exponentially as $s \to \infty$. What remains in order to apply Lemma \ref{le: MS10} is to show that $K_\o^{1/2} \notin L^1(a, \infty)$. Again, this is clear, since $K_\o$ grows exponentially.

We can therefore apply Lemma \ref{le: MS10} and conclude that there exist solutions $\b_\o^1, \b_\o^2: (a, \infty) \to \C$ to equation \eqref{eq: BetaODE} such that
\begin{align*}
  &\b_\o^1(s)e^{i\int_a^s K_\o(u)^{1/2}du}K_\o(s)^{1/4} \to 1, \\
  &\b_\o^2(s)e^{-i\int_a^s K_\o(u)^{1/2}du}K_\o(s)^{1/4} \to 1,
\end{align*}
which is equivalent to
\begin{align}
  &\b_\o^1(s)K_\o(s)^{1/4} - e^{-i\int_a^s K_\o(u)^{1/2}du} \to 0, \label{eq: LargeTimesBeta1} \\
  &\b_\o^2(s)K_\o(s)^{1/4} - e^{i\int_a^s K_\o(u)^{1/2}du} \to 0. \label{eq: LargeTimesBeta2}
\end{align}
It follows that $\b_\o^1$ and $\b_\o^2$ are linearly independent, therefore there exist constants $c_1, c_2 \in \C$ such that
\[
 \b_\o = c_1 \b_\o^1 + c_2 \b_\o^2.
\]
We now substitute back using $\a_\o(t) = \b_\o(\ln(t))$. Note that
\begin{align*}
 K_\o(\ln(t)) &= 4 \pi^2 \sum_{j=1}^3 {\o_j}^2 t^{2-2p_j}, \\
 \int_a^{\ln(t)} K_\o(u)^{1/2}du &= \int_{e^a}^t K_\o(\ln(u))^{1/2} \frac1u du \\
  &= 2 \pi \int_{e^a}^t \left( \sum_{j=1}^3 \frac{{\o_j}^2}{u^{2p_j}} \right)^{1/2} du.
\end{align*}
Note that we can, by change of the parameters $c_1, c_2$ if necessary, choose $a = \ln(t_0)$. Inserting this in equations \eqref{eq: LargeTimesBeta1} and \eqref{eq: LargeTimesBeta2} implies the theorem.
\end{proof}

\section{Application: Redshift in Kasner spacetimes} \label{sec: Redshift}

Light rays in general relativity are modelled by lightlike geodesic. It is well-known that the wavelength of a light ray changes under the impact of a gravitational field. This phenomenon is called \emph{redshift}. In general relativity, one usually models the wavelength of a lightlike geodesic as being inverse proportional to the energy of the lightlike geodesic modelling the light ray. The question for this section is, can we instead model the wavelength of light using a mode of the wave equation? In Minkowski space, this is true, since the Maxwell equations for light simplify to the wave equation. One important difference between Kasner spacetimes and Minkowski spacetimes is the `Big Bang` at $t=0$. Therefore, one can only expect a well-defined notion of wavelength, modeled by the wave equation, for large times. We will use a notion of `large time wavelength`, simply the inverse of the `large time frequency` defined in the previous section. It is shown that the redshift one obtains using this definition of wavelength coincides with the classical notion of redshift.

Let us start by fixing a lightlike geodesic. We specify the initial data for the lightlike geodesic on the Cauchy hypersurface $\{t_0\}\times \R^3$. We want to consider a light ray $\g$ sent out at time $t_0$ at an arbitrary point $x_0$ in space with the initial (spatial) direction given by an arbitrary vector $v$. 

\subsubsection*{Fixing a lightlike geodesic}

Let $\g: I \subset \R \to M$ be the future pointing lightlike geodesic, written in coordinates as
\[		
  \g(s) = (t(s), x^1(s), x^2(s), x^3(s)),
\]
with the initial data
\begin{align*} 
  &\g(s_0) = (t_0, x_0), \\
  &\g'(s_0) = \left(\sum_{j=1}^3{v_j}^2{t_0}^{2p_j}, v \right) \in T_{(t_0, x_0)}M,
\end{align*}
for fixed $x_0, v \in \R^3$. Note that $\g$ is indeed lightlike. Since $\d_j$ for $j=1,2,3$ are Killing vector fields, the linear momentum of $\g$ is conserved,
\[
	\ldr{\d_j, \g'(s)} = {t_0}^{2p_j}v_j=: a_j \in \R.
\]

\subsubsection*{The two notions of wavelength}

Let us first present the usual way of describing wavelength of a lightlike geodesic. The redshift using this definition of wavelength is called the \emph{cosmological redshift} and is presented for Robertson-Walker spacetimes in \cite[p. 353]{O'Neill}. We generalize this classical notion to Kasner spacetimes by defining the wavelength analogously.

\begin{definition}[Wavelength inverse proportional to the energy] \label{def: FrequencyCosmological}
Let $\g: I \to M$ be the lightlike geodesic we fixed above. The \emph{wavelength using the energy} $\lambda_\g^E: I \to \R_+$ of $\g$ is defined as 
\[
  \lambda_\g^E(s) := \frac h{E_\g(s)},
\]
for all $s \in I$, where $h$ is the Planck constant and $E_\g(s)$ is the energy measured by $\d_t$.
\end{definition}

We now turn to our model of the wavelength of $\g$ as a mode $\a_\o$ of the wave equation. A natural choice of $\o$ is the linear momentum of $\g$, i.\,e.\,
\[
	\o_j := \ldr{\d_j, \g'(s)} = a_j.
\]

\begin{definition}[The large time wavelength] \label{def: LargeTimeWavelength}
Let $\g: I \subset \R \to M$ be the lightlike geodesic we fixed above. The \emph{large time wavelength} $\lambda_\g^{LT}: I \to \R_+$ of $\g$, measured by $\d_t$ is defined as 
\[
  \lambda_\g^{LT}(s) := \frac1{f_{a}(t(s))},
\]
for all $s \in I$, where $t: I \to \R_+$ is the time component of $\g$ and $f_{a}(t)$ is the large time frequency (see Definition \ref{def: LargeTimeFrequency}) of the mode $\a_{a}$, where $a$ is the (conserved) linear momentum of $\g$.
\end{definition}

\subsubsection*{Comparison of the obtained redshifts}

\begin{definition}[Redshift of the lightlike geodesic $\g$.] \label{def: Redshift}
Let $\g: I \subset \R \to M$ be the lightlike geodesic we fixed above. Let $p$ and $q$ be in the image of $\g$, and let $s_p < s_q \in I$ be such that $\g(s_p) = p$ and $\g(s_q) = q$.  Assume that the wavelength of $\g$ observed by $\d_t$ is given by a function $\lambda_\g: I \to \R_+$. Then the redshift between $p$ and $q$ relative to the observer field $\d_t$ is defined by
\[
  z_\g(p, q) := \frac{\lambda_\g(s_q) - \lambda_\g(s_p)}{\lambda_\g(s_p)}.
\]
\end{definition}

The redshift is independent of whether we choose to model the wavelength using the energy or the large time wavelength:

\begin{thm} \label{thm: Redshift}
Let $(M,g)$ be a Kasner spacetime and let $\g$ be the above fixed geodesic with initial spatial direction $v$. Let $p$ and $q$ be in the image of $\g$, and let $s_p < s_q \in I$ be such that $\g(s_p) = p$ and $\g(s_q) = q$. The redshift obtained by using $\lambda^{LT}_\g$ as wavelength coincides with the redshift obtained by using $\lambda^E_\g$ as wavelength and equals
\[
  z_\g(p, q) = \left(\frac{ \sum_{j=1}^3 \frac{{a_j}^2}{{t(s_p)}^{2p_j}}}{\sum_{j=1}^3 \frac{{a_j}^2}{{t(s_q)}^{2p_j}}}\right)^{1/2} - 1,
\]
where $t: I \to \R_+$ is the time coordinate of $\g$ and $a$ is the (conserved) linear momentum of $\g$.
\end{thm}

\begin{proof}
Theorem \ref{thm: LargeTimes} implies that
\[
  f_{a}(t) = \left(\sum_{j=1}^3 \frac{{a_j}^2}{t^{2p_j}} \right)^{1/2},
\]
and therefore
\[
\lambda_\g^{LT}(s) = \frac1 {f_{a}(t(s))} = \frac1{\left(\sum_{j=1}^3 \frac{{a_j}^2}{t(s)^{2p_j}} \right)^{1/2}}.
\]
Inserting this into Definition \ref{def: Redshift} proves the first part. By definition of the geodesic being lightlike,
\begin{align*}
  0 =& \ldr{\g'(s), \g'(s)} \\
  =& -\left( t'(s)\right)^2 + \sum_{j=1}^3 (x^j)'(s) \ldr{\d_j, \g'(s)} \\ 
  =& -\left( t'(s)\right)^2 + \sum_{j=1}^3 \frac{\ldr{\d_j, \g'(s)}^2}{t(s)^{2p_j}} \\
  =& -\left( t'(s)\right)^2 + \sum_{j=1}^3 \frac{{a_j}^2}{t(s)^{2p_j}},
\end{align*}
for all $s \in I$. This implies that the energy is given by 
\[
  E_\g(s)= -\ldr{\g'(s), \d_t} = t'(s) = \left( \sum_{j=1}^3 \frac{{a_j}^2}{t(s)^{2p_j}} \right)^{1/2}  
\]
and hence
\[
 \lambda_\g^E(s) = \frac h {E_\g(s)}= \frac h {\left(\sum_{j=1}^3 \frac{{a_j}^2}{t(s)^{2p_j}} \right)^{1/2}},
\]
for all $s \in I$. Inserting into Definition \ref{def: Redshift} finishes the proof.
\end{proof}



\end{sloppypar}
\end{document}